\newif\ifconference
\newif\ifanonymous
\conferencefalse
\anonymousfalse
\newif\ifheapspublished
\heapspublishedfalse

\documentclass[twoside,leqno,final]{article}
\usepackage{geometry}
\usepackage{hyperref}
\usepackage[hyphenbreaks]{breakurl}
\usepackage{lmodern}
\usepackage[T1]{fontenc}

\usepackage{graphicx}

\usepackage[english]{babel}

\usepackage{amsmath}
\usepackage{amssymb}
\usepackage{amsthm}
\usepackage{graphicx}
\hypersetup{colorlinks=true, allcolors=blue, breaklinks=true}
\usepackage{float}
\usepackage[textwidth=20mm,obeyFinal,colorinlistoftodos]{todonotes}
\usepackage{lineno}

\setuptodonotes{size=\tiny}

\usepackage[capitalize,nameinlink,noabbrev]{cleveref}
\usepackage{enumitem}

\usepackage[backend=biber,style=alphabetic,sorting=nty,maxnames=10,useprefix]{biblatex}         
\let\citet\textcite
\usepackage[ruled,vlined]{algorithm2e}
\DontPrintSemicolon
\SetVlineSkip{1pt}
\NoCaptionOfAlgo
\SetKwIF{If}{ElseIf}{Else}{if}{:}{else if}{else:}{}%
\SetKwFor{For}{for}{:}{}%
\SetKwFor{While}{while}{:}{}%
\SetKwProg{Fn}{}{:}{}
\SetCommentSty{textrm}

\newtheorem{theorem}{Theorem}[section]

\newtheorem{lemma}[theorem]{Lemma}

\theoremstyle{definition}

\newtheorem*{theorem*}{Theorem}

\crefname{step}{Step}{Steps}
\crefname{algorithm}{Algorithm}{Algorithms}
\crefname{theorem}{Theorem}{Theorems}
\crefname{@theorem}{Theorem}{Theorems}
\crefname{proposition}{Proposition}{Propositions}
\crefname{observation}{Observation}{Observations}
\crefname{lemma}{Lemma}{Lemmas}
\crefname{claim}{Claim}{Claims}
\crefname{problem}{Problem}{Problems}
\crefname{conjecture}{Conjecture}{Conjectures}
\crefname{question}{Question}{Questions}
\crefname{example}{Example}{Examples}
\crefname{fact}{Fact}{Facts}    
\crefname{invariant}{Invariant}{Invariants}

\crefname{algocf}{Algorithm}{Algorithms} 
\Crefname{algocf}{Algorithm}{Algorithms}

\newcommand{\eps}{\epsilon}
\renewcommand{\lg}{\log}
\newcommand{\OPT}{\mathrm{OPT}}
\newcommand{\OO}{\mathrm{O}}
\newcommand{\el}{\ell}

\newcommand{\Insert}{\mbox{\it insert\/}}
\newcommand{\Findmin}{\mbox{\it find-min\/}}
\newcommand{\Deletemin}{\mbox{\it delete-min\/}}
\newcommand{\Decreasekey}{\mbox{\it decrease-key\/}}
\DeclareMathOperator*{\polylog}{poly\,log}
\newcommand{\citeauthornolink}[1]{{\NoHyper\citeauthor{#1}\endNoHyper}}

\makeatletter
\AtBeginDocument{
  \let\oldlog\log
  
  \def\log{\texorpdfstring{\log}{log }}
  \def\lg{\texorpdfstring{\lg}{log }}
  \hypersetup{
    pdftitle = {\@title},
    pdfauthor = {\ourauthors},
    pdfsubject = {\ourabstract}
  }
  \let\log\oldlog
  \let\lg\oldlg
}
\makeatother

\author{
Bernhard Haeupler\thanks{INSAIT, Sofia University ``St.~Kliment Ohridski'' \& ETH Zurich, \texttt{bernhard.haeupler@insait.ai}}
\and
Richard Hladík\thanks{ETH Zurich, \texttt{ethz@rihl.cz}}
\and
John Iacono\thanks{Université libre de Bruxelles, \texttt{ulb@johniacono.com}}
\and
Václav Rozhoň\thanks{Charles University, \texttt{vaclavrozhon@gmail.com}}
\and
Robert E.~Tarjan\thanks{Princeton University, \texttt{ret@cs.princeton.edu}}
\and
Jakub Tětek\thanks{\texttt{j.tetek@gmail.com}}
}
\date{}

\ifanonymous
\def\ourauthors{Anonymous authors}
\else
\def\ourauthors{Bernhard Haeupler, Richard Hladík, John Iacono, Václav Rozhoň, Robert E. Tarjan, Jakub Tětek}
\fi
\def\ourkeywords{sorting under partial information, partial orders, topological sort, working-set heaps}

\def\ourabstract{%
We consider the problem of sorting $n$ items
, given the outcomes of $m$ pre-existing comparisons.  We present a simple and natural deterministic algorithm that runs in $\OO(m+\log T)$ time and does $\OO(\log T)$ comparisons, where $T$ is the number of total orders consistent with the pre-existing comparisons. 
\texorpdfstring{\par\smallskip\par}{}%
Our running time and comparison bounds are best possible up to constant factors, thus resolving a problem that has been studied intensely since 1976 (Fredman, Theoretical Computer Science). The best previous 
algorithm with a bound of $\OO(\lg T)$ on the number of comparisons has a time bound of \texorpdfstring{$\OO(n^{2.5})$}{O(n\^{}2.5)} and is more complicated.
\texorpdfstring{\par\smallskip\par}{}%
Our algorithm combines three classic algorithms: topological sort, heapsort with the right kind of heap, and efficient search in a sorted list.  It outputs the items in sorted order one by one.  It can be modified to stop early, thereby solving the important and more general top-$k$ sorting problem: Given $k$ and the outcomes of some pre-existing comparisons, output the smallest $k$ items in sorted order.  The modified algorithm solves the top-$k$ sorting problem in minimum time and comparisons, to within constant factors. 
}

\begin{document}

\title{Fast and Simple Sorting Using Partial Information}

\def\ref#1{\textcolor{red}{[\textbackslash ref is disabled. Please use \textbackslash cref instead. Use \textbackslash cref\{sec:foo,sec:bar\} to reference two things at once.]}}

\maketitle

\begin{abstract}
\ourabstract

\end{abstract}


\section{Introduction}
We consider the problem of sorting $n$ totally ordered items using a minimum number of binary comparisons, given the outcomes of $m$ pre-existing comparisons. This problem has been called \emph{sorting under partial information}~\cite{kahn-kim-supi-1992}, and it has been intensively studied since 1976~\cite{fredman-generalized-supi-1976}.

We present a simple and natural deterministic algorithm for this problem that runs in $\OO(m+\lg T)$ time and uses $\OO(\log T)$ comparisons, where $T$ is the number of total orders consistent with the pre-existing comparison outcomes. These bounds are optimal up to constant factors. The best previous algorithm with an $\OO(\log T)$ bound on comparisons requires $\OO(n^{2.5})$ time and is more complicated. An independent and concurrent work~\cite{vanderhoog2024tight} implies an $\OO(n^{\omega})$-time algorithm, where $\omega$ is the matrix multiplication exponent, the smallest number $\omega$ such that $n$ by $n$ matrices can be multiplied in $\OO(n^{\omega})$ time. Currently the lowest known bound for $\omega$ is $\omega \leq 2.371339$ \cite{matrix-multiplication-sota}. 

Our algorithm for this problem combines three classic algorithms in a natural way: topological sort, heapsort with the right kind of heap, and efficient searching in a sorted list. Unlike many previous algorithms, ours does not use any estimate of $T$ or any measure of entropy to determine the next comparison.  It uses two simple data structures: an array and a pairing heap, a simple kind of self-adjusting heap.

Our approach is closely related to a recent result of \citet{dijkstra-universal-optimality} proving that Dijkstra's algorithm with an appropriate heap is universally optimal for the task of sorting vertices by their distance from a source vertex.  Our problem places weaker constraints on the required heap efficiency, allowing us to use a pairing heap rather than the custom-designed heap that gives the Dijkstra result.

Our algorithm outputs the items in sorted order one-by-one.  It can be modified to efficiently solve the important and more general top-$k$ sorting problem: Given $k$ and the outcomes of some pre-existing comparisons, output the smallest $k$ items in sorted order.  The modified algorithm solves the top-$k$ sorting problem in minimum time and comparisons, to within constant factors.  In this more general problem we require more of the heap, but a variant of the pairing heap that we call the \emph{double pairing heap} suffices.  

Specifically, we prove the following two theorems:

\begin{theorem}
\label{thm:main_intro_1}
There is a simple and natural deterministic algorithm that, given the outcomes of 
$m$ pre-existing comparisons between pairs of $n$ totally ordered elements, outputs all the elements in sorted order in $\OO(m+\OPT)$ time using only $\OO(\OPT)$ additional comparisons.
\end{theorem}

\begin{theorem}\label{thm:main_intro_2}
There is a simple and natural deterministic algorithm that, given the outcomes of 
$m$ pre-existing comparisons between pairs of $n$ totally ordered elements and a parameter $k$, outputs the smallest $k$ elements in sorted order in $\OO(m+\OPT(k))$ time using only $\OO(\OPT(k))$ additional comparisons.
\end{theorem}

\paragraph{Roadmap}
In \Cref{S:preliminaries}, we give a graph-theoretic formulation of the sorting problems we consider, and we define and discuss heaps and their working sets.  A heap with a working-set bound is a key ingredient in our algorithms. 
In \Cref{S:prior-work} we discuss related work. This includes work on the sorting problems we consider and on other related sorting and ordering problems, work on sampling and counting topological orders, and work on heaps with a working-set bound. 
In particular, we discuss an alternative sorting algorithm developed by~\citet{vanderhoog2024} in recent follow-on work to the original version of our paper~\cite{haeupler-et-al2024}, version 1, April 6, 2024.  In the current version of our paper we have used the key lemma of~\citet{vanderhoog2024} to simplify our analysis. 

In \Cref{S:topological}, we present our basic algorithm, which we call \emph{topological heapsort}. We prove that it runs in $\OO(m+\lg T)$ time and does $\OO(n+\lg T)$ comparisons.  We eliminate the additive $n$ term in the number of comparisons in \Cref{S:topological-heapsort-lookahead} by adding an additional step to our algorithm, producing an algorithm that we call \emph{topological heapsort with lookahead}.  This algorithm is best possible to within constant factors in both its running time and its number of comparisons.  Our algorithm uses two simple data structures: an array and a pairing heap.  In \cref{S:top-k-sorting}, we adapt both versions of our algorithm to efficiently find the smallest $k$ items in sorted order, for any (given) $k$.
In \Cref{sec:sampling}, we describe how to use our sorting algorithm and any algorithm that generates a uniformly random topological order to produce a crude approximation to the number of possible topological orders.
In \Cref{S:double-pairing-heaps}, we develop and analyze a new heap implementation, the \emph{double pairing heap}, a variant of the pairing heap that has the efficiency needed for the results in \Cref{S:top-k-sorting}.
In~\cref{S:alternative}, we discuss the alternative sorting algorithm of~\citet{vanderhoog2024} and its relation to our algorithms.

\section{Preliminaries}\label{S:preliminaries}
In this section we present a graph-theoretic formulation of the problems we consider, and we define and discuss heaps with the working set bound, a critical component of our algorithms.

\subsection{Graph-Theoretic Formulation}\label{S:graph-formulation}

To allow a fine-grained analysis of the computational complexity of algorithms for the two sorting problems we consider, it is convenient to formulate these problems in the following graph-theoretic way: 

The input is given in the form of a directed graph $G$ whose $n$ vertices are the items, with each of the $m$ pre-existing comparison outcomes $v < w$ between two items $v$ and $w$ represented by an arc $vw$.  Our first sorting problem, which we call the \emph{DAG sorting problem}, is to compute the unknown total order of the vertices of $G$ by doing additional binary comparisons of the vertices. This problem has a solution if and only if $G$ is a directed acyclic graph (DAG), and each possible solution is a \emph{topological order} of the DAG, i.e., a total order such that if $vw$ is an arc, then $v < w$.  

A lower bound on the number of comparisons needed to solve the problem is $\log T$,\footnote{Throughout this paper ``$\log$'' without a base denotes the base-two logarithm.} where $T$ is the number of possible topological orders of $G$.  This lower bound is an immediate consequence of the well-known lower bound for sorting by binary comparisons, and more generally for identifying one object in a set by asking yes-no questions: A simple decision-tree argument shows that the number of queries needed is at least the base-2 logarithm of the size of the set, and this bound holds to within an additive constant even in the average case.  Furthermore the lower bound holds even for randomized algorithms, both Las Vegas and Monte Carlo, by Yao's min-max principle~\cite{Yao77}. 

Our primary objective is to minimize the number of comparisons; our secondary objective is to make the running time small. The problem is interesting for settings in which comparisons are expensive and also when $m$ is small compared to $n\log n$; otherwise, one can ignore the pre-existing comparisons and sort the items from scratch. 

An equivalent view of the problem is that the given comparisons define a partial order, and the problem is to find an unknown linear extension of this partial order by doing additional comparisons. We prefer the DAG view, because the DAG induced by the pre-existing comparisons may not be transitively closed (if $uv$ and $vw$ are arcs then so is $uw$) nor transitively reduced (if $vw$ is an arc then there is no other path from $v$ to~$w$).  Ideally we want a solution that works for an arbitrary set of pre-existing comparisons, and we do not want to take the time to compute either the transitive reduction or the transitive closure of the DAG representing these comparisons.


Our second sorting problem, which we call the \emph{top-$k$ DAG sorting problem}, generalizes DAG sorting.  Given a DAG representing a set of pre-existing comparison outcomes and a parameter $k \leq n$, the problem is to output the smallest $k$ items in increasing sorted order.  We discuss this problem in more detail in \Cref{S:top-k-sorting}.  

\subsection{Heaps and Their Working Sets}\label{S:heaps-working-set}

Our algorithm substantially deviates from the techniques previously used for DAG sorting. Our basic algorithm is a classical topological sorting algorithm modified to use a heap to choose the next-smallest vertex.  We prove that if the heap has the so-called \emph{working-set bound}, then our DAG sorting algorithm is efficient.

A \emph{heap} is a data structure $H$ storing a set of items, each having a key selected from a totally ordered set.  The heap is initially empty.  For our purpose, heaps support three operations:  $\Findmin(H)$, which returns an item of smallest key in heap $H$; $\Insert(H, x)$, which inserts item $x$, having a predefined key and not in $H$ already, into heap $H$; and $\Deletemin(H)$, which deletes and returns $\Findmin(H)$.  Some heap implementations support a fourth operation, $\Decreasekey(H, x, k)$, which, given the location of an item $x$ in heap $H$ whose key is greater than $k$, replaces the key of $x$ in heap $H$ by $k$.  Fibonacci heaps~\cite{fredmanfibonacci1987} and other equally efficient heap implementations support find-min in $\OO(1)$ worst-case time and no comparisons, insert and decrease-key in $\OO(1)$ amortized time, and delete-min on an $n$-item heap in $\OO(\log n)$ time.

In our application, we need a more refined bound for delete-mins that depends on the sequence of $\Insert$ and $\Deletemin$ operations done on the heap.  This is the \emph{working-set bound}, defined as follows.  We view each item that is deleted and later re-inserted as being a new item, so that each item is inserted and deleted only once.  Consider an item $x$ that is inserted into the heap at some time and later deleted.  An item $y$ is in the \emph{working set} of $x$ if it is inserted into the heap at or after the time $x$ is inserted but before $x$ is deleted.  (Thus $x$ is in its own working set.)  The \emph{working-set size} $w(x)$ of $x$ is the size of its working set.  The heap has the \emph{working-set bound} if each find-min takes $\OO(1)$ time and no comparisons, each insertion takes $\OO(1)$ time, and each delete-min of an item $x$ takes $\OO(1+\log w(x))$ time.  These bounds can be amortized.  (If decrease-key is a supported operation, it has no required time bound.)  

The working-set bound is related to a similar bound for binary search trees, one possessed by splay trees~\cite{splay-trees}, but it is not the same.  There are several alternative definitions of the working-set bound for heaps, of which the one we have given here is the weakest. \citet{iacono-pairing-heaps} considered a stronger bound, which he called just the working-set bound but we shall call the \emph{intermediate working-set bound}.  The \emph{intermediate working-set size} $w'(x)$ of an item $x$ is the maximum number of items in its working set that are in the heap at the same time, with the maximum taken over all times from the insertion of $x$ to just before the deletion of $x$.  The heap has the \emph{intermediate working-set bound} if each insertion takes $\OO(1)$ time and each delete-min of an item $x$ takes $\OO(1+\log w'(x))$ time.  These bounds can be amortized.  Iacono proved that a pairing heap~\cite{pairing-heaps}, a simple form of self-adjusting heap, has the intermediate working-set bound, provided that the heap ends empty.  By Iacono's result, a pairing heap has the efficiency we need for DAG sorting.  

To solve the top-$k$ DAG sorting problem, we need a bound better than the working-set bound.  This is the \emph{strict working-set bound}, which is defined in the same way as the working-set bound except that only items that are eventually deleted from the heap are members of a working set: Items inserted into the heap but never deleted do not count.  The strict bound differs from the non-strict bound only if some items are inserted but never deleted; that is, the heap ends non-empty.  In \cref{S:double-pairing-heaps}, we develop and analyze the \emph{double pairing heap}, a new form of pairing heap in which delete-mins do two pairing passes rather than just one, and we prove that these heaps have the strict intermediate working-set bound, which is defined analogously to the intermediate working-set bound but only counting items that are eventually deleted.  A double pairing heap has the efficiency we need for top-$k$ DAG sorting.  This is the only heap implementation we know has the strict working-set bound.  In particular, we do not know whether this bound holds for standard pairing heaps.  We leave this question as an open problem.    

\section{Related Work}\label{S:prior-work}

\paragraph{DAG sorting and top-$k$ DAG sorting}
\citet{fredman-generalized-supi-1976} considered the generalization of the DAG sorting problem in which there is an unknown total order selected from an \emph{arbitrary} subset of the total orders of $n$ elements, and the problem is to find the unknown total order by doing binary comparisons of the elements.  He showed that if there are $T$ possible total orders, $\log T + 2n$ binary comparisons suffice. His algorithm is highly inefficient, however.

For the special case of sorting a DAG,~\citet{kislitsin} conjectured that there always exists a \emph{balancing comparison}: a comparison ``$x < y$?'' such that the fraction of the topological orders of $G$ for which the answer is ``yes'' lies between $1/3$ and $2/3$.  This is the 1/3--2/3 conjecture.  More generally, one can ask whether there is always a comparison that splits the set of topological orders into two subsets, each with a fraction of between $\delta$ and $1-\delta$ of the original set, for some constant $\delta \leq 1/3$.  (In general there is no such comparison for any value of $\delta$ larger than $1/3$.)  \citet{fredman-generalized-supi-1976} and \citet{1/3-conjecture-2} later independently made the 1/3-2/3 conjecture. \citet{posets-3/11} showed that the general conjecture is true for $\delta = 3/11$.  The value of $\delta$ has been improved several times since then~\cite{balancing-extensions-1,balancing-extensions-2,balancing-extensions-3}. It follows immediately that there is an algorithm that always does at most $\log_{1/(1-\delta)} T$ comparisons: Repeatedly find a balancing comparison and do it.

\citet{1/3-conjecture-2} proved that for the special case of merging two sorted lists, the 1/3--2/3 conjecture is true, and gave a way to find a 1/3--2/3 balancing comparison in polynomial time.  For this special case, \citet{HwangLin} showed earlier that $\OO(\log T)$ comparisons suffice, and \citet{BrownTarjan} gave an algorithm running in $\OO(\log T)$ time.  Our new algorithm merges \emph{any} collection of sorted lists in minimum time and comparisons, to within constant factors. 

\citet{kahn-kim-supi-1992} gave a polynomial-time algorithm for the DAG sorting problem that does $\OO(\lg T)$ comparisons.  Their algorithm uses the ellipsoid method to find comparisons that are good in the amortized sense.  \citet{cardinal-supi-2010} gave several algorithms for sorting a DAG that avoid the use of the ellipsoid method.  In particular, they devised an algorithm that runs in $\OO(n^{2.5})$ time and does $\OO(\lg T)$ comparisons. This is the fastest previous algorithm that sorts a DAG in $\OO(\log T)$ comparisons.  Another algorithm in \cite{cardinal-supi-2010} has the same time bound and does $(1+\eps)\log T  + \OO_\eps(n)$ comparisons, where $\eps$ is an arbitrarily small positive constant, and the constant in the big $\OO$ term depends on epsilon.  

Independent of and concurrent with our work, \citet{vanderhoog2024tight} considered the DAG sorting problem, but with a different representation of the input. They assume that the transitive closure of the DAG (the partial order implied by the DAG) is represented by a matrix whose entries can be queried one at a time. For any positive constant $c$ they gave an algorithm that solves the sorting problem in $\OO(n^{1+1/c})$ matrix queries, $\OO(c\log T)$ comparisons, and $\OO(n^{1+1/c}+c\log T)$ time.  They also showed that in this query model these bounds are best possible.  

Computing the transitive closure matrix is closely related to binary matrix multiplication and can be done in $\Theta(n^{\omega})$ time, where $\omega$ is the matrix multiplication exponent. The result of \citeauthornolink{vanderhoog2024tight} therefore implies an $\OO(n^{\omega}+\log T)=\OO(n^{2.371339})$-time algorithm with $\OO(\log T)$ comparisons for the DAG sorting problem, improving on the $\OO(n^{2.5})$-time algorithm of \cite{cardinal-supi-2010}. 


Another variant of great practical importance is the top-$k$ DAG sorting problem, in which only the $k$ smallest vertices must be returned in sorted order, where $k$ is an input parameter. This problem has many practical applications in rank aggregations, where often only the ranking of the top candidates is of interest, and in settings where comparisons are expensive, because human experts are required to rank alternatives via comparisons. An important and typical such setting is reinforcement learning with human feedback~\cite{christiano2017deep}. See~\citet{top-k-supi} for further pointers and applications of this problem. \citet{top-k-supi} give an algorithm for top-$k$ DAG sorting that runs in $\OO(n \log k + m)$ time and does $\OO(n \log k)$ comparisons.  They also provide empirical evaluations. From a theoretical perspective we remark that the top-$k$ DAG sorting problem can be solved from scratch (ignoring the pre-existing comparisons) in $\OO(n+k\log k)$ time and comparisons, by finding the $k$ smallest vertices using linear-time selection~\cite{BlumFPRT73}, and then sorting these $k$ vertices. In \cref{S:top-k-sorting} we show that our sorting algorithm easily adapts to do top-$k$ DAG sorting in minimum time and comparisons, to within constant factors, assuming a suitable representation of the DAG and using a suitable heap implementation. This natural extension is unique to our approach: Prior algorithms for DAG sorting do not generalize naturally or in some cases at all to the top-$k$ setting.   

For additional related work, see the survey of \citet{survey-supi-and-similar-problems}.

\paragraph{An Alternative Algorithm}
In follow-on work to the original version~\cite{haeupler-et-al2024} of our paper,~\citet{vanderhoog2024} have proposed an alternative algorithm for the DAG sorting problem with the same asymptotic efficiency.  
While we dub our algorithm explained in \cref{S:topological} \emph{topological heapsort}, their algorithm might well be called \emph{topological insertion sort}. We discuss this alternative algorithm and its connection to our algorithm more closely in \cref{S:alternative}. 

In this revision of our paper~\cite{haeupler-et-al2024}, we have simplified our efficiency analysis by basing it on Lemma 3 of~\cite{vanderhoog2024}. We have included a simple, self-contained proof of a version of their lemma that suffices for us.

In \cref{S:topological-heapsort-lookahead}, we improve our basic algorithm so that it has not only the best possible time complexity, but also the best possible query complexity. There are two very similar ways of  doing this. One of them, \emph{handling a longest path separately}, we used in the original version of our paper~\cite{haeupler-et-al2024}, version 1, April 6, 2024. This idea dates back to \citet{cardinal-supi-2010}; \citet{vanderhoog2024} independently used a similar approach. The other one, \emph{handling bottlenecks separately}, we used in our paper~\cite{haeupler-et-al2024}, version 2, July 22, 2024. These two ideas are interchangeable in the sense that they can be used both in our algorithm and in the algorithm of~\citet{vanderhoog2024}.  In this version of our paper we use the bottlenecks idea. 

\paragraph{Heaps with a working-set bound}\label{P:working-set-bound}

\citet{elmasry-strong-working-set} devised a heap with a bound better than the intermediate working-set bound: The time for a delete-min of $x$ is $\OO(1+\log s(x))$, where $s(x)$ is the number of items in the working set of $x$ that are still in the heap just before $x$ is deleted.  
\citet{dijkstra-universal-optimality} developed a heap that has the intermediate working-set bound and also supports decrease-key operations in $\OO(1)$ amortized time.  Both these data structures support find-min in $\OO(1)$ worst-case time and no comparisons.  Neither is known to have the strict working-set bound.

See~\citet{kozma_saranurak,munro2019dynamic,risaSurveyPic,bernhardsSurvey} for other beyond-worst-case time bounds for heaps. 

\paragraph{Related sorting and ordering problems}
Our approach is strongly influenced by a recent result~\cite{dijkstra-universal-optimality} on the distance-ordering problem. The input to this problem is a directed graph with non-negative arc weights and a source vertex. The problem is to sort the vertices of the input graph by their distance from the source. This problem can be solved in $\OO(m + n\log n)$ time and comparisons by running Dijkstra's algorithm using a Fibonacci heap~\cite{fredmanfibonacci1987}.  These bounds are best possible in the worst case.  The authors of \cite{dijkstra-universal-optimality} improve this worst-case result by giving a so-called universally optimal algorithm: They show that if Dijkstra's algorithm is implemented using a heap with the working-set bound that also supports decrease-key operations in $\OO(1)$ amortized time, then for any fixed graph $G$ with a given source vertex, the algorithm takes the minimum time needed (to within a constant factor) to solve the problem on $G$ with a worst-case choice of arc weights. Moreover, an extension of Dijkstra's algorithm minimizes not only the time but also the number of comparisons to within a constant factor.

The DAG sorting problem also asks for a universally optimal algorithm in that on any given DAG it minimizes the running time and number of comparisons.  Both problems are generalizations of sorting.  Not only are the two problems similar, but we show that they can be solved by similar techniques.

Another algorithm related to ours is the adaptive heapsort algorithm of \citet{levcopoulos1993adaptive_heapsort}. This algorithm sorts a sequence of numbers by first building a heap-ordered tree, the \emph{Cartesian tree} of the sequence, defined recursively as follows: The root is the minimum number in the sequence, say $x$, its left subtree is the Cartesian tree of the prefix of the sequence up to but not including $x$, and its right subtree is the Cartesian tree of the suffix of the sequence from the element after $x$ to the end of the sequence.  The Cartesian tree can be built in at most $2n-3$ comparisons~\cite{levcopoulos1993adaptive_heapsort}.  The algorithm finishes the sort using a heap to store the possible minima, initially only the root.  After a delete-min, say of $v$, the children of $v$ in the Cartesian tree are inserted into the heap.  \citet{levcopoulos1993adaptive_heapsort} use a standard heap in their algorithm, but if one uses a  heap with the working-set bound instead, our results imply that adaptive heapsort does $\OO(n+\log T)$ comparisons, where $T$ is the number of topological orders of the Cartesian tree viewed as a DAG.

An ordering problem that is dual to DAG sorting is that of producing a given partial order on a totally ordered set, by doing enough comparisons so that the partial order induced by the comparison outcomes is isomorphic to the input partial order.  The two problems are dual in the sense that if an $n$-element partial order has $T$ topological orders, the partial-order production problem on the same partial order requires $\log n! - \log T$ comparisons~\cite{Schonhage76, Yao89}, both worst-case and expected-case.  \citet{CardinalFJJM09} gave an algorithm for partial-order production that does a number of comparisons within a lower-order term of the lower bound plus $\OO(n)$.  

\paragraph{Sampling and counting topological orders}
Even though we give a sorting algorithm that does $\OO(\lg T)$ comparisons, calculating $\lg T$ exactly is hard, since the problem of determining $T$ is \#P-complete \cite{counting-T-is-hard}. There are algorithms that compute $T$ approximately \cite{counting-linear-extensions-1991,counting-linear-extensions-2017}, however. \citet{counting-T-is-hard} have shown that a constant-factor approximation to $T$ can be obtained with high probability using $\OO(n^2 \polylog n)$ calls to an oracle that generates a uniformly random topological order. There are multiple results on generating a random topological order \cite{random-linear-extension-1991,random-linear-extension-1991-2,random-linear-extension-1998,random-linear-extension-2006}, with the state-of-the-art approach having $\OO(n^3\log n)$ time complexity and leading to an $\OO(n^5\polylog n)$ approximation algorithm for~$T$. In \cref{sec:sampling}, we show how to use one oracle call to get a constant-factor approximation to $\lg T$, and hence a polynomial approximation to $T$.

\section{Topological Heapsort and Its Efficiency}\label{S:topological}
This section introduces our basic algorithm, which we call \emph{topological heapsort}. The algorithm combines two classic algorithms, topological sort \cite{knuth1997art} and heapsort \cite{heapsort1,heapsort2} in a simple way. We describe the algorithm in \cref{S:topological-heapsort}.  In \cref{S:topological-heapsort-analysis}, we prove that our algorithm runs in $\OO(m+\log T)$ time and does $\OO(n+\log T)$ comparisons, if it is implemented with a heap having the working-set bound if it ends empty, and in particular with a pairing heap. 

\subsection{Topological Heapsort}
\label{S:topological-heapsort}

We start by recalling a basic result in graph theory: A directed graph is a DAG if and only if it has a topological order.  One can prove this by the following classic \emph{topological sort} algorithm~\cite{kahn1962topological,knuth1997art}:  Call a vertex a \emph{source} if it has no entering arcs.  Given a directed graph, repeatedly delete a source and its outgoing arcs, until there are either no vertices or no sources.  In the former case, the vertex deletion order is a topological order; in the latter case, one can find a cycle by starting at any remaining vertex and building a path by repeatedly traversing some arc in the reverse direction and continuing until a vertex is repeated.

To make this algorithm efficient, one needs to keep track of the current set of sources.  The algorithm of \citet{kahn1962topological} does this by maintaining for each vertex its current in-degree (number of incoming arcs). The sources are the vertices with in-degree zero.  When a vertex is deleted, the in-degrees of its immediate successors (those reached by an outgoing arc) are decremented.  In each iteration, the algorithm finds a source by examining all the remaining vertices, which results in a worst-case time bound of $\OO(n^2)$.  \citet{knuth1997art} added the idea of maintaining the current set of sources in a separate data structure, for which he uses a queue.  This reduces the running time to $\OO(m+n)$. 

Given a DAG, our task is to find a \emph{specific} topological order, the one corresponding to the unknown total order of the vertices.  Our basic algorithm is topological sort with the current set of sources stored in a heap.  The key of a vertex is the vertex itself.  Each step deletes the minimum vertex, say $v$, from the heap, adds $v$ to the total order, decrements the in-degree of each vertex $w$ such that $vw$ is an arc, inserts into the heap each such $w$ whose in-degree is now zero, and finally deletes $v$ and its outgoing arcs from the DAG. The complete description of the algorithm is given in \cref{alg:topological-heapsort}. 

This algorithm is not only a version of topological sort, it is also a form of heapsort~\cite{heapsort1,heapsort2}, which in turn is a form of selection sort: The algorithm adds the vertices to the total order in increasing order, using a heap to do so.  It differs from standard heapsort in that the heap contains only the current sources, which are the only candidates for the next minimum, rather than all the undeleted vertices.  We call the algorithm \emph{topological heapsort}.  

\begin{algorithm}[H]
\SetAlgoLined
\KwData{A directed acyclic graph $G$}
\KwResult{A sorted list of vertices $L$}
 Initialize $L$ as an empty list to store sorted vertices\;
 Compute the in-degree of each vertex\;
 Initialize heap $H$ to contain all vertices with in-degree zero (the sources)\;
 \While{$G$ is not empty}{
  Delete the minimum vertex $v$ from $H$\;
  Add $v$ to the back of $L$\;
  
  \ForEach{arc $vw$}{
   Decrement the in-degree of $w$\;
   \If {$w$ now has in-degree zero (is a source)}{
    Add $w$ to $H$\;
   }
  }
  Delete $v$ and each outgoing arc $vw$ from $G$\;
 }
 \caption{\cref{alg:topological-heapsort}: Topological Heapsort}
\label{alg:topological-heapsort}
\end{algorithm}

A proof by induction on the number of deleted vertices shows that the $k^\mathrm{th}$ deleted vertex is the $k^\mathrm{th}$ smallest: Given that the $k-1$ smallest vertices have been deleted from the heap and from the graph, the $k^\mathrm{th}$ smallest must be a source.  Since the heap contains all the current sources, the $k^\mathrm{th}$ delete-min returns the $k^\mathrm{th}$ smallest vertex.  Thus topological heapsort is correct.   Its running time is $\OO(m+n)$ plus the time required for $n$ insertions into $H$ intermixed with $n$ delete-mins from $H$.  All the vertex comparisons are done by the heap insertions and deletions, so to bound the number of comparisons it suffices to bound the time of these operations.  
The input can be a list of the arcs in $G$; if it is, as part of the initialization we build for each vertex a list of its outgoing arcs.

We implement topological heapsort using a heap that has the working-set bound if it ends empty, such as a pairing heap.

\subsection{Efficiency of Topological Heapsort}\label{S:topological-heapsort-analysis}

To prove that topological heapsort is efficient, we must estimate the number of comparisons needed to sort the vertices of a given DAG $G$.  This has to be at least $\log T$ in the worst case and on the average, and even for Las Vegas and Monte Carlo randomized algorithms, where $T$ is the number of topological orders of $G$, by the standard information-theory lower bound argument for sorting by binary comparisons.  The lower bound for randomized algorithms follows from the classic work of Yao~\cite{Yao77}.  

We shall prove the following theorem: 

\begin{theorem}\label{T:topological-heapsort-bound} 
Topological heapsort runs in $\OO(m+\log T)$ time and does $\OO(n+ \log T)$ comparisons, if the heap has the working-set bound provided that it ends empty.
\end{theorem}

In stating bounds, we assume that $T \geq 2$, which implies $n \geq 2$; otherwise, there is a unique topological order, and no comparisons are needed to find it. 

As a first, simple observation, we prove that $n = \OO(m+\log T)$.  This allows us to eliminate an additive $n$ term in the running time of the algorithm.    

\begin{lemma}\label{L:n-small}
We have $n = \OO(m+\log T)$.
\end{lemma}
\begin{proof}
Assume $m <n-1$; otherwise the lemma is immediate.  The initial number of sources is at least $n-m$.  Any permutation of these sources can be extended to a topological order.  Hence $T \geq (n-m)!$, which implies $\log T \geq (n-m)/2$, since $n-m \geq 2$.
\end{proof}  

Recall from~\cref{S:heaps-working-set} that $w(v)$, the working set size of $v$, is the number of vertices inserted into the heap from the time $v$ is inserted until $v$ is deleted, including $v$. 
We need a bound on the sum of $\log w(v)$ over all the vertices.  The main lemma of~\cite{vanderhoog2024} provides a way to obtain such a bound.  We need some additional concepts.  For integers $i$ and $j$ such that $i \leq j$, we denote by $[i, j]$ the interval of integers $\{i, i+1,\dots j\}$.  Let $\{[a_i,b_i]\mid 1 \leq i \leq k\}$ be a set of $k$ intervals of integers, each a subset of $[1,k]$.  We define a DAG $I$ whose vertices are the intervals $\{[a_i,b_i]\}$ and whose arcs are the pairs $[a_i,b_i],[a_j,b_j]$ such that $b_i < a_j$.  We call $I$ the DAG \emph{associated with} the set of intervals $\{[a_i,b_i]\}$   We denote by $T(I)$ the number of topological orders of $I$.

\begin{lemma}\label{L:interval-bound-0}\textup{\cite{vanderhoog2024}}
Let $I$ be the DAG associated with the set of intervals $\{[a_i, b_i]\mid 1 \leq i \leq k\}$.  Then $\sum_{i=1}^k \log (b_i-a_i+1)=\OO(\log T(I))$.
\end{lemma}

For completeness we state and prove a variant of this lemma that suffices for us.

\begin{lemma}\label{L:interval-bound}
Let $I$ be the DAG associated with the set of intervals $\{[a_i, b_i] \mid 1 \leq i \leq k\}$.  Then $\sum_{i=1}^k \log (b_i-a_i+1) \leq \log T(I)+ k\log e$, where $e$ is the base of the natural logarithm.  \end{lemma}
\begin{proof}
For each $i$ between $1$ and $k$ inclusive, choose a real number $r_i$ uniformly at random from the real interval $(0, k]$, independently for each $i$.  With probability $1$ the $r_i$ are distinct.  Let $\pi$ be the permutation of $[a_i,b_i]$ given by ordering the $[a_i,b_i]$ in increasing order on $r_i$.  Each possible permutation is equally likely.  If each $r_i$ is in the real interval $(a_i-1, b_i]$, then $\pi$ is a topological order of $I$.  The probability of this happening is $\prod_{i=1}^k (b_i-a_i+1)/k$.  It follows that $T(I)\geq k! \cdot \prod_{i=1}^k (b_i-a_i+1)/k$.  Taking logarithms gives $\log T(I) \geq \sum_{i=1}^k \log (b_i-a_i+1) + \log k! - k\log k$.  By Stirling's approximation of the factorial, $\log k! \geq k\log k - k\log e$.  The lemma follows.   
\end{proof}

We remark that the bound in~\cref{L:interval-bound} is an equality for the set of $n$ intervals all equal to $[1,n]$.  Hence the constant factors in this form of the bound cannot be improved.   

We use~\cref{L:interval-bound} to bound the sum of the logarithms of the working-set sizes in a run of topological heapsort. 

\begin{lemma}\label{L:working-set-log-T}
Given a run of topological heapsort, let $w(v)$ for $v \in V$ be the working-set size of $v$.  Then $\sum_{v\in V}\log w(v) \leq  \log T+n\log e$.  
\end{lemma}
\begin{proof}
Consider a run of topological heapsort.  Let $[v_1, v_2,\dots, v_n]$ be the sequence of vertices in the order they are inserted into the heap.  For each vertex $v_i$, let $[a_i=i, b_i]$ be the set of indices of vertices inserted into the heap from the time $v_i$ is inserted until the time $v_i$ is deleted, including $i$, the index of $v_i$.  Then $w(v_i)=b_i-a_i+1$.  Let $I$ be the DAG associated with the set of intervals $\{[a_i,b_i]\mid 1 \leq i \leq n\}$.  We claim that any topological order of $I$ gives a topological order of $G$ if we replace each interval $[a_i,b_i]$ in the order by $v_i$.  Indeed, if there is an arc from $v_i$ to $v_j$ in $G$, then $v_i$ must be deleted from the heap before $v_j$ is inserted, so $b_i < j$, which implies that $[a_i,b_i]$ precedes $[a_j,b_j]$ in the topological order of $I$.  Each distinct topological order of $I$ gives a distinct topological order of $G$.  Hence $T(I) \leq T$.  The lemma follows from~\cref{L:interval-bound}.
\end{proof}

Now we can prove \cref{T:topological-heapsort-bound}.  The running time of topological heapsort is $\OO(m+n)$ plus the time to do the heap operations, which do all the comparisons.  The time to do the heap operations is $\OO(n+\log T)$ by \cref{L:working-set-log-T} and the definition of the working-set bound, as is the number of comparisons. By \cref{L:n-small}, $n= \OO(m+\log T)$.  It follows that the total running time is $\OO(m+\log T)$ and the total number of comparisons is $\OO(n+\log T)$.  

\section{Topological Heapsort with Lookahead}\label{S:topological-heapsort-lookahead}

The bound on comparisons in \cref{T:topological-heapsort-bound} includes an additive term linear in $n$.  This term is significant if the number of topological orders of the DAG is small, specifically sub-exponential in $n$.  In this section, we augment topological heapsort to eliminate this term in the bound.

The first step is to determine when the additive $n$ term is significant.  To do this we introduce two more concepts.  For each vertex $v$, we define $\el(v)$, the \emph{level} of $v$, to be the number of vertices on any longest path (one having the most vertices) to $v$ in the original DAG $G$.  A path to $v$ of $\el(v)$ vertices starts at a source of $G$; a vertex $v$ is a source if and only if $\el(v)=1$. We define a \emph{bottleneck} to be a vertex that is the only vertex on its level.

Bottlenecks are useful to us because if $v$ is a bottleneck, every vertex on a level higher than $\el(v)$ must follow $v$ in any topological order.  This implies that in every topological order, the bottlenecks are in increasing order by level.  This means that we can solve the sorting problem by sorting the bottlenecks by level (which takes no comparisons), sorting the non-bottlenecks by value, and merging the two sorted lists.  This is our approach.  The details require a bit of care.

We shall prove that $\log T$ is small compared to $n$ (and the original algorithm would thus perform $\Theta(\log T + n) = \omega(\log T)$ comparisons) only when at least a positive constant fraction of the vertices are bottlenecks.  As we shall see, this implies that the number of comparisons needed to sort the non-bottlenecks and to insert them into the sorted list of bottlenecks is $\OO(\log T)$, as desired.    

We can compute vertex levels in $\OO(m+n)$ time by initializing $\el(v)=1$ for each vertex $v$ and running topological sort with the following addition: Set $\el(w) \gets \max\{\el(w), \el(v)+1\}$ each time an arc $vw$ is deleted (when $v$ is the source being deleted).  Given vertex levels, we can find all the bottlenecks and order them in increasing order both by level and by value in $\OO(n)$ time and no comparisons.  

\begin{lemma}\label{L:many-bottlenecks}
If $G$ has $b$ bottlenecks, then $\log T \geq (n-b)/2$.
\end{lemma}
\begin{proof}
Let $V_i$ be the set of vertices on level $i$, and let $\ell$ be the number of levels.  Since every level that does not contain a bottleneck contains at least two vertices, $2(\ell-b) \leq n-b$.  This gives $\ell \leq (n+b)/2$.  Every arc is from a vertex on a lower level to one on a higher level.  Thus we can form a topological order of $G$ by ordering the vertices level by level in increasing level order, and within each level ordering the vertices arbitrarily.  It follows that\footnote{Throughout this paper, $|S|$ denotes the size of a set $S$.} \[T \geq \prod_{i=1}^{\ell}|V_i|! \geq \prod_{i=1}^{\ell}2^{|V_i| - 1}= 2^{n-\ell}.\]
Taking logarithms gives $\log T \geq n - \el$.  Combining this with the inequality $\ell \leq (n+b)/2$ gives the lemma.
\end{proof}

By \cref{L:many-bottlenecks}, topological heapsort does $\OO(\lg T)$ comparisons when run on any DAG with at most $(1-\epsilon)n$ bottlenecks, where $\epsilon$ is any positive constant, since in this case $\log T \ge \epsilon n/2 = \Omega(n)$.  Thus we only need a way to handle graphs in which at least a large constant fraction of the vertices are bottlenecks.

To handle such graphs, we find all the bottlenecks and sort them as already described.  This takes no comparisons.  Then we run topological heapsort on the entire graph, but we only insert the non-bottlenecks into the heap $H$.  We use a separate data structure, specifically an array, $B$, to store the set of undeleted bottlenecks in order by level.  The next source to be deleted from the graph is either the smallest vertex in $H$ or the first vertex in $B$.  Since the set of bottlenecks is much larger than the set of non-bottlenecks, to make the algorithm efficient we delete bottlenecks in groups.  We would like to use the following idea: If $v$ is the smallest vertex in $H$, find the largest  vertex in $B$ less than $v$, say $u$, and delete from the graph all bottlenecks on $B$ up to and including $u$, doing the deletions in order by level.  This does not quite work, because deleting a bottleneck source from $G$ can create new non-bottleneck sources, and such a new source could be smaller than the smallest remaining old source.  But we can make the idea work by dividing the list of bottlenecks into \emph{chunks} in such a way that only deletion of the last vertex in a chunk can create a new non-bottleneck source, and handling the chunks one at a time.

Let us expand on this idea.  We break the sorted list of bottlenecks into chunks as follows: For each non-bottleneck $v$, mark the bottleneck $w$ of highest level such that there is an arc $wv$. There are thus at most $n - b$ marked bottlenecks.
Observe that the deletion of an unmarked bottleneck cannot create a new non-bottleneck source: The only possible source its deletion can create is the next larger bottleneck.  We define a chunk to be a maximal sequence of bottlenecks in order by level containing at most one marked bottleneck, the last one.  This definition partitions the list of bottlenecks into at most $n-b+1$ chunks.  Constructing the list of bottlenecks and dividing it into its chunks takes $\OO(m+n)$ time and no comparisons.

Our chunk-based algorithm works as follows: We initialize an array $B$ to contain the first chunk of bottlenecks.  We run topological heapsort but insert only the non-bottleneck sources into the heap $H$.  To find the current smallest vertex when both $B$ and $H$ are non-empty, we compare the smallest vertex in $H$, say $v$, with the last vertex in $B$, say $w$.  If $v > w$, we delete all remaining vertices in $B$ from $B$ and from $G$, add them to the sorted list of vertices $L$ under construction, insert any new non-bottleneck sources into $H$, and refill $B$ with the next chunk, if any.  The only such bottleneck whose deletion can create a new source is the last one, and only if it is marked.  If on the other hand $v < w$, we find all vertices in $B$ less than $w$, all of which are unmarked.  To find such bottlenecks, we do an exponential/binary search in $B$.  We delete each such vertex from $B$ and from $G$ and add it to $L$.  Once all such bottlenecks are processed, we delete $v$ from $H$ and from $G$, add $v$ to $L$, and insert any new non-bottleneck sources into $H$.  

We call the resulting algorithm \emph{topological heapsort with lookahead}. In full detail, the algorithm is as follows: Find the bottlenecks and their levels.  For each non-bottleneck vertex $w$, mark the bottleneck $v$ of the highest level such that $vw$ is an arc, if any.  Partition the bottlenecks into chunks, one ending at each marked vertex and possibly a last one containing the unmarked vertices after the last marked one.  Initialize an array $B$ to contain the first chunk of bottlenecks, in increasing order by level.  Initialize the heap $H$ to contain each source of $G$ except the lowest-level bottleneck, if any.  Initialize a sorted list of vertices $L$ to empty.  Repeat the applicable one of the following cases until $G$ is empty:

\begin{itemize}

\item [] Case 1: $H$ is non-empty, and either $B$ is empty or its minimum-level vertex is greater than $\Findmin(H)$.  Set $v \gets \Deletemin(H)$.
Delete $v$ from $G$, add $v$ to $L$, and insert any new non-bottleneck sources into $H$.

\item [] Case 2: $B$ is non-empty, and either $H$ is empty or its smallest vertex is larger than the maximum-level vertex in $B$.  Delete each vertex in $B$, add these vertices in increasing order by level to $L$, delete them from $G$, and insert any new non-bottleneck sources into $H$.  Then refill $B$ with the next chunk, if any.

\item [] Case 3: Cases 1 and 2 do not apply; that is, $H$ and $B$ are non-empty and the smallest vertex in $H$ is smaller than the largest vertex in $B$.  Set $v \gets \Findmin(H)$.  Find the largest vertex $x$ in $B$ that is less than $v$.  To find $x$, compare $v$ to the first, second, fourth, eighth, \dots{} vertex in $B$ until finding two consecutively compared vertices $y$ and $z$ in $B$ such that $y < v < z$.  Then do a binary search on the set of vertices in $B$ between $y$ and $z$ to find $x$.  If $x$ is the $j$-th vertex in $B$, this search takes $\OO(1+\log j)$ time and $\OO(1+\log j)$ comparisons.  Process the bottlenecks $u$ in $B$ up to and including $x$ in increasing order by level.  To process a bottleneck $u$, add it to $L$ and delete it from $B$ and from $G$.

\end{itemize}

We shall show that topological heapsort with lookahead is correct and efficient on any graph. As we have shown, the added complication of handling the bottlenecks separately is unnecessary unless most of the vertices are bottlenecks.

\begin{theorem}
Topological heapsort with lookahead is correct.
\end{theorem}
\begin{proof}
Suppose $B$ is non-empty at the beginning of an iteration of the main loop of the algorithm.  Then $B$ is non-empty at the end of the iteration, with one exception: The iteration executes Case 2, and at the end of the iteration all bottlenecks have been deleted from $G$.  It follows by induction on the number of iterations that $B$ is non-empty at the beginning of each iteration until all bottlenecks are deleted from $G$, and after the last deletion of a bottleneck from $G$, $B$ remains empty.

Suppose that up to the beginning of some iteration, the algorithm has correctly deleted the smallest vertices from $G$ in increasing order.  By the invariant in the previous paragraph, the smallest undeleted vertex at the beginning of the iteration is either the smallest vertex in $H$ or the lowest-level vertex in $B$.  It follows that if the iteration is an application of Case 1, it correctly deletes the smallest undeleted vertex. Suppose it is an application of Case 2 or Case 3.  Then the lowest-level vertex in $B$ is the smallest undeleted vertex, and if this vertex is unmarked, its deletion does not create any new non-bottleneck sources.  Hence each successive vertex deleted from $B$ during the application of the case is the smallest undeleted vertex.  In Case 2, the last deletion of a vertex from $B$ can be of a marked bottleneck.  This deletion can create new non-bottleneck sources, but these must be bigger than the marked bottleneck just deleted.  It follows by induction that the algorithm deletes vertices from $G$ in increasing order, and hence is correct.       
\end{proof}

As already mentioned, one can find the bottlenecks in $G$ by doing a topological sort, computing the level of each vertex $v$ by using the recurrence $\ell(v)=\max(\{1\}\cup\{\ell(u)+1\mid uv ~\textup{is an arc of } G\})$, and then finding each vertex that is the only vertex on its level.  This takes $\OO(m+n)$ time.

One can mark the bottlenecks by processing them in decreasing order by level.  First, unmark all the vertices.  Then, for each bottleneck $v$ in decreasing order by level, if there is an arc $vw$ to a non-bottleneck $w$ and $w$ is unmarked, mark both $v$ and $w$.  The marking process takes $\OO(m+n)$ time and no comparisons, and uses only the outgoing arc lists of the bottlenecks.  Once marking is finished, the marks of the non-bottleneck vertices can be ignored.

The total running time, excluding the time spent on heap operations and searches of $B$ in Case 3, is $\OO(m+n)= \OO(m+\log T)$.  Let $b$ be the number of bottlenecks.  The number of non-bottleneck vertices is $n-b$.  The number of marked bottlenecks is at most $n-b$.  Each iteration of Case 1 deletes one non-bottleneck from $H$ and from the graph and does one comparison outside of the delete-min from $H$.  Each iteration of Case 2 except the last one adds one marked bottleneck to $B$ and does one comparison.  Thus the number of comparisons not done during heap operations and searches of $B$ is at most $2(n-b)=\OO(\log T)$ by~\cref{L:many-bottlenecks}.  The number of heap insertions is $n-b$.  The amortized time and number of comparisons required for these is $\OO(n-b)=\OO(\log T)$.  It remains to bound the time and comparisons required by the heap deletions and the searches in $B$.  The next two lemmas bound these quantities.

\begin{lemma}\label{L:lookahead-deletions}
The time and number of comparisons required for the heap deletions is $\OO(\log T)$.
\end{lemma}
\begin{proof}
Let $G'$ be the graph formed from the transitive closure of $G$ by deleting all the bottlenecks.  Let $T'$ be the number of topological orders of $G'$.  Each topological order of $G'$ can be extended to one of $G$ by merging this order with the list of bottlenecks in increasing order by level while preserving topological order.  Thus $T' \leq T$.  The heap operations done during a run of topological heapsort with lookahead on $G$ are exactly the same as those done by the corresponding run of topological heapsort on $G'$.  The amortized time and comparisons required for the heap deletions done by the latter are $\OO(n-b+\log T')=\OO(\log T')$ by~\cref{L:working-set-log-T} and~\cref{L:many-bottlenecks}.  The lemma follows.
\end{proof}

\begin{lemma}\label{L:lookahead-search-time}
In a run of topological heapsort with lookahead, the searches of $B$ take $\OO(\log T)$ time and $\OO(\log T)$ comparisons.  
\end{lemma}
\begin{proof}
For each non-bottleneck $v$, let $B(v)$ be the set of unmarked bottlenecks deleted from $B$ during the iteration of Case 3 preceding the iteration of Case 1 that deletes $v$; if an iteration of Case 1 or 2 precedes the iteration of Case 1 that deletes $v$, let $B(v)$ be the empty set.  If $B(v)$ is non-empty, the time and number of comparisons spent determining $B(v)$ in Case 3 is $\OO(\log|B(v)|+1)$.  The sets $B(v)$ are disjoint, and if non-bottleneck $v$ is less than non-bottleneck $w$, then all vertices in $B(v)$ are less than all vertices in $B(w)$.   Since $v$ is a source of G when it is added to $H$, which is before the iteration of Case 3 preceding the deletion of $v$, there is no arc from a vertex in $B(v)$ to $v$.  Since each vertex in $B(v)$ is less than $v$, there is no arc from $v$ to any vertex in $B(v)$.  It follows that we can obtain a valid topological order by ordering all the bottlenecks in increasing order by level and inserting each non-bottleneck vertex $v$ into $B(v)$ in any of the $|B(v)|+1$ possible ways, including just before the first vertex on $B(v)$ and just after the last one.  (If $B(v)$ is empty, there is only one possible insertion position, just after the vertex added to the sorted list $L$ just before $v$.)
This gives $T \geq \prod_v (|B(v)|+1)$.  Hence $\log T \geq \sum_v \log(|B(v)|+1)$. By~\cref{L:many-bottlenecks}, the time and number of comparisons spent on searches of $B$ is $\OO(\log T)$. 
\end{proof}

Combining our bounds, we obtain~\cref{thm:main_intro_1}: Topological heapsort with lookahead sorts the vertices of a DAG $G$ in $\OO(m+\lg T)$ time and $\OO(\lg T)$ comparisons.

We remark that if the problem must be solved repeatedly for a fixed DAG with different total orders, the bottlenecks need only be found and marked once. 

\section{Top-\emph{k} DAG Sorting}\label{S:top-k-sorting}

In this section, we show that with small modifications our algorithms efficiently find the smallest $k$ vertices in sorted order, given $k$.  Our formulation of this problem is the same as for the complete sorting problem: The input is given in the form of an acyclic directed graph $G$ whose vertices are the items, having an arc $vw$ for each given comparison outcome $v < w$.  We denote by $T_k$ the number of possible outcomes: These are all sequences of $k$ items that are compatible with the given comparison outcomes and such that if $vw$ is a DAG arc and $w$ is in the sequence, so is $v$.  Equivalently, these are the sequences that can be generated by running the topological ordering algorithm until it orders $k$ vertices, at each step choosing the next source in all possible ways.  We denote by $m_k$ the number of DAG arcs out of the smallest $k-1$ vertices, and by $\OPT(k)$ the minimum number of comparisons needed to solve the problem for a given $k$.  In stating time bounds we assume $T_k \geq 2$, which implies $n \geq 2$; otherwise, solving the problem requires no comparisons.

Unlike the complete sorting problem, in which the information-theory lower bound of $\log T$ on comparisons is tight to within a constant factor, for the top-$k$ DAG sorting problem the corresponding lower bound of $\log T_k$ is not necessarily tight.  Let $S_k$ be the set of sources after the $k-1$ smallest sources are deleted. Then $|S_k|-1$ comparisons are needed to determine the $k^\mathrm{th}$ smallest vertex, even given the $k-1$ smallest.  Thus $\OPT(k) \geq \max\{\log T_k, |S_k|-1\}$.  If $\log T_k$ is small, the second term in this bound can dominate.

 We shall show that topological heapsort, slightly modified, solves the problem in $\OO(m_k+\OPT(k))$ time and $\OO(k+\OPT(k))$ comparisons, and then that topological heapsort with lookahead, also slightly modified, solves the problem in $\OO(m_k+\OPT(k))$ time and $\OO(\OPT(k))$ comparisons.  Each of these results requires an appropriate input representation.  If the DAG is given only as a list of its arcs, both algorithms require $\OO(m+n)$ additional preprocessing time but no additional comparisons.  This bound is $\OO(m+\log T)$ by \cref{L:n-small}.  If the problem is to be solved repeatedly on the same DAG, the preprocessing needs to be done only once.  
One can solve the problem in $\OO(n+k\log k)$ time and comparisons by ignoring the pre-existing comparisons and solving the problem from scratch: Find the $k$ smallest vertices using linear-time selection~\cite{BlumFPRT73} and then sort them.  Our results are meaningful in situations where $m_k+\OPT(k)$ or $m+\OPT(k)$ (depending on the input representation) is small compared to $n+k\log k$, or comparisons are very expensive.  An important application in which comparisons are expensive is in recommendation systems, as discussed in~\cite{top-k-supi}: Given some partial information about user preferences, the system asks a user or group of users to make additional comparisons, in order to determine the top $k$ recommendations.  The same situation occurs in the tuning of large language models~\cite{top-k-supi}.

\subsection{Top-\emph{k} DAG Sorting via Topological Heapsort}\label{S:top-k-topological-heapsort}

We begin by adapting topological heapsort to solve the problem.  We assume that we are given the initial set of sources, the initial in-degree of each vertex, and the set of arcs exiting each vertex.  If the input does not provide this information, we pre-compute it from the set of pre-existing comparisons, which takes $\OO(m+n)=\OO(m+ \log T)$ time but no comparisons.  Spending this amount of time is unavoidable if, for example, the input is just a list of the pre-existing comparison outcomes in no particular order.

We run topological heapsort using a heap that has the strict working-set bound, such as a double pairing heap.  We stop the algorithm just before the $k^\mathrm{th}$ delete-min.  At this time the algorithm has found the $k-1$ smallest vertices, in increasing order.  We find the $k^\mathrm{th}$ smallest by doing a find-min on the heap, which takes $\OO(1)$ time and no comparisons.

Let us analyze the algorithm.  The number of insertions into the heap is $k-1+ |S_k|$, where  $S_k$ is the set of sources after the deletion of the smallest $k-1$ vertices, as defined above.  The amortized time of these insertions is $\OO(k+ |S_k|-1)$.

To bound the time of the $k-1$ heap deletions, let $G(k)$ be the subgraph induced by the $k$ smallest vertices: $G(k)$ contains these vertices and all arcs of $G$ with both ends among the $k$ smallest vertices.  Let $T(k)$ be the number of topological orders of $G(k)$.  Then $T(k) \leq T_k$, since every topological order of $G(k)$ is a possible output of the top-$k$ DAG sorting algorithm when run on $G$ to find the $k$ smallest vertices.  The heap deletions done by this algorithm are exactly the same as the first $k-1$ deletions if topological heapsort is run on $G(k)$, although the former algorithm may do additional insertions. 
If the heap has the strict working-set bound, then these additionally inserted vertices are never deleted, and hence do not count in the bound.  It follows that the strict working-set bound for the heap deletions done by the top-$k$ DAG sorting algorithm is identical to the working-set bound for the heap deletions done if topological heapsort is run to completion on $G(k)$.  
By the results of \cref{S:topological-heapsort}, the total amortized time of the heap deletions done by the top-$k$ DAG sorting algorithm is $\OO(k+\log T(k))=\OO(k+\log T_k)$.  Adding the bounds for the insertions and deletions and using the inequality $\OPT(k) \geq \max\{\log T_k, |S_k|-1\}$, we find that the algorithm does $\OO(k+\OPT(k))$ comparisons.  Its running time (excluding the initialization time if initialization is needed) is $\OO(k+ m_k + \OPT(k))$.  The proof of \cref{L:n-small} extends in a straightforward way to show that $k=\OO(m_k+\log T_k)$, so the running time bound is $\OO(m_k + \OPT(k))$.

In the special case of no pre-existing comparisons, topological heapsort merely inserts all $n$ vertices into the heap and then does $k$ deletions.  If the heap has the strict working-set bound, the running time is $\OO(n+ k\log k)$, optimum in both time and comparisons.  This algorithm is arguably simpler than using selection and sorting.  We thank one of the anonymous referees for this observation.

\subsection{Top-\emph{k} DAG Sorting via Topological Heapsort with Lookahead}

We can eliminate the additive $k$ term in the number of comparisons needed for finding the smallest $k$ vertices by using topological heapsort with lookahead.   We assume as in~\cref{S:top-k-topological-heapsort} that we are given the initial set of sources, the initial in-degree of each vertex, and the set of arcs exiting each vertex. If not, we pre-compute this information from the set of pre-existing comparisons, which takes $\OO(m+n)=\OO(m+ \log T)$ time but no comparisons.  In addition, we assume we are given a list of the bottlenecks in increasing order by level, up to and including the bottleneck on the highest level $l_k$ such that levels $1$ through $l_k$ contain fewer than $k$ vertices in total.

If the required bottlenecks are not provided, we compute them as follows: Run a topological sort that chooses sources in non-decreasing order by level.  It suffices to use a queue containing the current sources, as in Knuth's version of topological sort.  At each step, delete the first vertex on the queue, say $v$, delete $v$ and all arcs $vw$ from the graph, and for each such arc $vw$ do the following updates: Set $\el(w) \gets \max\{\el(w), \el(v)+1\}$, decrement the in-degree of $w$, and add $w$ to the back of the queue if it is now a source.  At all times the queued vertices are on at most two levels, and all the lower-level vertices precede all the higher-level ones, as one can easily prove by induction.  It follows that the vertices are deleted both in a topological order and in non-decreasing order by level.  Stop just after deleting the $k^\mathrm{th}$ vertex, say $v$, from the queue.  (There is no need to delete $v$ and its outgoing arcs from the graph.)  The needed bottlenecks are those on levels less than that of $v$.   

We  mark the given (or computed) bottlenecks as in~\cref{S:topological-heapsort-lookahead}.  Specifically, unmark all these bottlenecks, and for each such bottleneck $v$ and each arc $vw$, unmark $w$.  Process the given bottlenecks in decreasing order by level.  To process a bottleneck $v$, if $v$ has an outgoing arc to an unmarked vertex $w$ that is not one of the given bottlenecks, mark $v$ and $w$.

Once the given bottlenecks are marked, we run topological heapsort with lookahead using a heap with the strict working-set bound and treating the set of given bottlenecks as the entire set of bottlenecks.  We run this algorithm until it has deleted $k-1$ vertices from the graph, and then do a find-min on the heap to find the $k^\mathrm{th}$ smallest vertex.

This algorithm is clearly correct.  If the required bottlenecks are not provided as part of the input, the time to compute them is at most linear in the number of sources of $G$ plus the number of arcs out of the first $k-1$ vertices deleted from the queue.  We denote the number of such arcs by $m'_k$. The value of $m'_k$ can be larger or smaller than $m_k$, the number of arcs out of the \emph{smallest} $k-1$ vertices.  In any case, finding the needed bottlenecks takes $\OO(k+m'_k+|S_k|)=\OO(m'_k+\OPT(k))$ time and no comparisons.  

Let $v$ be the highest-level bottleneck among the given ones, and suppose it is the $j$-th smallest vertex.  Until the algorithm deletes $v$ from $G$, only vertices on levels no greater than the level of $v$ are sources.  Thus $j < k$.  This implies that all the given bottlenecks are among the smallest $k-1$ vertices.  It also implies that the time to mark the given bottlenecks is $\OO(m_k)$.  Also, because the heap has the strict working-set bound, until the algorithm deletes $v$ from $G$ it does the same computations and has the same bound as the algorithm in~\cref{S:topological-heapsort-lookahead} run to completion on the subgraph induced by the $j$ smallest vertices.  Each topological order on this subgraph is a prefix of a possible solution to the top-$k$ DAG sorting problem.  Hence the time for this part of the computation is $\OO(m_k + \OPT(k))$, and the number of comparisons is $\OO(\OPT(k))$.  The remainder of the computation runs the algorithm of~\cref{S:top-k-topological-heapsort}, since all the given bottlenecks have been deleted.  By the results of~\cref{S:top-k-topological-heapsort}, this part of the computation takes $\OO(m_k + \OPT(k))$ time and $\OO(k-j+\OPT(k))$ comparisons.  Let $S$ be the set of the first $k$ vertices in non-decreasing order by level, with ties broken arbitrarily.  Since there are at least $k-j-1$ non-bottlenecks among the vertices in $S$, and since any topological order of the subgraph induced by the vertices in $S$ is a possible solution to the top-$k$ DAG sorting problem, $k-j = \OO(\OPT(k))$ by~\cref{L:many-bottlenecks}.  We conclude that the algorithm runs in $\OO(m_k +\OPT(k))$ time and does $\OO(\OPT(k))$ comparisons.  If the needed bottlenecks must be computed, $\OO(m'_k)$ additional time is required for this purpose.

This gives us~\cref{thm:main_intro_2}: Topological heapsort with lookahead finds the smallest $k$ vertices of a DAG $G$ in sorted order in $\OO(m+\OPT(k))$ time and $\OO(\OPT(k))$ comparisons.

\section{Acknowledgments}
BH, RH, VR, and JT were partially funded by the Ministry of Education and Science of Bulgaria's support for INSAIT as part of the Bulgarian National Roadmap for Research Infrastructure. 
BH and RH were partially funded through the European Research Council (ERC) under the European Union's Horizon 2020 research and innovation program (ERC grant agreement 949272).
VR was partially funded through the European Research Council (ERC) under the European Union's Horizon 2020 research and innovation program (ERC grant agreement 853109).
RH and JT were supported by the VILLUM Foundation grant 54451. 
Part of this work was done while JT was working at and RH was visiting BARC at the University of Copenhagen. RH would like to thank Rasmus Pagh for hosting him there. Part of this work was done while JT was working at INSAIT.
JI was supported by the Fonds de la Recherche Scientifique-FNRS.
RT's research at Princeton was partially supported by a gift from Microsoft.  Part of this work was done during RT's visits to INSAIT and to the Simons Institute for the Theory of Computing.

\printbibliography

\appendix
\newpage
\section{Sampling and Counting Topological Orders}
\label{sec:sampling}
Given a DAG $G$ and its corresponding number of topological orders $T$, our algorithm yields a simple way of estimating the value of $\log T$ to within a constant factor. The idea is that the DAG sorting problem with an unknown total order selected uniformly at random takes $\Omega(\log T)$ comparisons with high probability.  Thus if the algorithm is run on one sample selected uniformly at random, we obtain a good approximation of $\log T$.

\begin{theorem}
Let $G$ be a directed acyclic graph with $T$ topological orders. Assume that there is an algorithm that returns a topological order $\OO(1)$-pointwise close to uniform\footnote{We say that a distribution $p$ is $c$-pointwise close to $q$ if for every element $x$ we have $q_x / c \le p_x \le c\cdot q_x$. } in time $t_{sample}$. 
Then there is an algorithm that runs in time $\OO(t_{sample} + \log T)$, performs  $\OO(\log T)$ comparisons
, and returns a constant-factor approximation of the value of $\log T$ with error probability $\OO(1/T^{0.9})$.
\end{theorem}
\begin{proof}
The algorithm merely samples a topological order, runs topological heapsort with lookahead on the sample, and returns the number of comparisons made by the algorithm.  The running time of the algorithm is $\OO(n+m+\log T + t_{sample})$. Since in order to sample the order, we have to read the whole input DAG, this is equal to the desired $\OO(t_{sample} + \log T)$. The number of comparisons it does is $\OO(\log T)$ (note that we do not use any comparisons when generating the sample).

The number of comparisons $X$ done by topological heapsort with lookahead is $X = \OO(\log T)$ by \cref{thm:main_intro_2}.  Thus our algorithm returns an upper bound on $\log T$ that is at most a constant factor larger than the true value. 

It remains to give a similar lower bound.
We prove that with high probability $X = \Omega(\log T)$. Suppose $\mathcal O$ is a sample $c$-pointwise close to uniform. Consider the event $\mathcal{E}$ that $X \leq \frac{\log (T/c)}{10}$. For any topological order $\mathcal O'$, we have
\[
P[\mathcal O = \mathcal O' \mid \mathcal{E}] = \frac{P[\mathcal O=\mathcal O' \;\cap\; \mathcal{E}]}{P[\mathcal{E}]} \leq\frac{P[\mathcal O=\mathcal O']}{P[\mathcal{E}]} \leq \frac{c/T}{P[\mathcal{E}]}. 
\]
The conditional entropy of $\mathcal O$ is then
\[
H(\mathcal O \mid \mathcal{E}) =\sum_{\mathcal O' \in \mathcal{E}} P[\mathcal O=\mathcal O' \mid \mathcal{E}] \cdot \log \frac{1}{P[\mathcal O=\mathcal O' \mid \mathcal{E}]} \geq \log \left(\frac{T \cdot P[\mathcal{E}]}{c}\right).
\]

The comparisons performed by topological heapsort with lookahead uniquely determine the topological order~$\mathcal O$. Thus by Shannon's source coding theorem for symbol codes \cite{shannon1948mathematical}, we have $E[X \mid \mathcal{E}] \geq H(\mathcal O \mid \mathcal{E})$, and we can write
\[
\frac{\log(T/c)}{10} \geq E[X \mid \mathcal{E}] \geq H(\mathcal O \mid \mathcal{E}) \geq \log \left(\frac{T \cdot P[\mathcal{E}]}{c}\right) \,.
\]
Solving for $P[\mathcal{E}]$, we get
\[
P[\mathcal{E}] \leq \left(\frac{c}{T}\right)^{9/10} ,
\]
which concludes the proof.
\end{proof}

\section{Double Pairing Heaps}\label{S:double-pairing-heaps}
In this section, we describe pairing heaps and double pairing heaps and prove that double pairing heaps have the strict intermediate working-set bound, which implies that they have the strict working-set bound.  Thus the double pairing heap is a suitable heap implementation to use in our algorithms for the top-$k$ DAG sorting problem.  Our terminology for heaps is similar to that of \citet{sinnamon2023efficiency}.  We ignore the $\Decreasekey$ operation since our algorithms do not use it.

We do not know whether pairing heaps themselves, or in general any heap implementation other than double pairing heaps, has the strict working-set bound.  We leave open the exploration of this topic.  In our analysis of double pairing heaps, we indicate where a similar analysis breaks down for pairing heaps. 

\subsection{Design of Double Pairing Heaps}

We represent a heap by a rooted tree whose nodes are the heap items.  The tree is heap-ordered: If a node $x$ is the parent of a node $y$, the key of $x$ is no greater than the key of $y$.  Access to the heap is via the tree root, making $\Findmin$ an $\OO(1)$-time and zero-comparison operation.

We manipulate heap-ordered trees by doing \emph{links} and \emph{cuts}.  A link of the roots of two item-disjoint trees combines the two trees into one by making the root of the larger key a child of the root of the smaller key, breaking a tie arbitrarily.  The root of the new tree is the \emph{winner} of the link; the new child is the \emph{loser} of the link.  We use the term ``link'' to refer both to the operation of doing a link and to the resulting parent-child pair.  We denote by $xy$ a link won by $x$ and lost by $y$.

We maintain the children of a node in a list ordered by linking time, latest first (leftmost).  Thus the loser of a link becomes the new first (leftmost) child of its new parent.  We represent a heap using two pointers per node, to its first child and to its next sibling.

A cut of a link $xy$ breaks the link, thereby breaking the tree containing $x$ and $y$ into two trees, one rooted at $y$ containing all descendants of $y$, and the other rooted at the old root, containing all non-descendants of $y$ (including $x$).

To insert a new item $x$ into a heap, we make $x$ into a one-node heap and link its root, $x$, with the root of the existing heap; if the existing heap is empty, the heap rooted at $x$ replaces the empty heap.  To do a $\Deletemin$, we cut the links between the heap root, say $x$, and its children.  This makes the list of children of $x$ into a list of roots.  We repeatedly link pairs of these roots until only one root remains.  If the initial heap contained only $x$, the new heap is empty.

The only flexibility in this implementation is in the choice of which links to do during a $\Deletemin$. A \emph{pairing heap} does these links in two passes.  The first pass, the \emph{pairing pass}, links the roots on the list in consecutive pairs first-to-last, the first with the second, the third with the fourth, and so on, leaving one root unlinked if the original number is odd.  The second pass, the \emph{assembly pass}, repeatedly links the last root on the list with the next-to-last root.  When a pairing or assembly link is done, the winner retains its current position on the list of roots; the loser is no longer a root and is deleted from the list.  A \emph{double pairing heap} does two left-to-right pairing passes (instead of just one) followed by a right-to-left assembly pass.

\subsection{Analysis of Double Pairing Heaps}

To analyze double pairing heaps we need some additional terminology.  Consider an arbitrary sequence of intermixed inserts and delete-mins on an initially empty heap.  We call a node \emph{temporary} or \emph{permanent}, respectively, if it is eventually deleted by a delete-min or not.  After the last operation, the heap contains all the permanent nodes and none of the temporary nodes.

We call a link \emph{real} if it is between two temporary nodes and \emph{phantom} otherwise.  A link between a temporary node and a permanent node cannot be won by the permanent node, because the temporary node is eventually deleted, which cannot happen until its parent is deleted, which in turn cannot happen if the parent is permanent.  It follows that every phantom link is lost by a permanent node.

We call a link an \emph{insertion}, \emph{pairing}, or \emph{assembly} link if it is done by an insertion, a pairing pass, or an assembly pass, respectively.  Omitting the $\Findmin$ operations (each of which takes $\OO(1)$ time and no comparisons), the time for a sequence of heap operations starting with an empty heap is bounded by a constant times one plus the number of links.  Thus it suffices to bound the number of links.  Indeed, it is enough to bound the number of real second-pass pairing links, as we now show:

\begin{lemma}\label{L:double-link-bound}
The number of links done in a sequence of double pairing heap operations starting with an empty heap is at most five per insertion and four per real pairing link done in the second pairing pass of a delete-min.
\end{lemma}

\begin{proof}
There is at most one insertion link per insertion.  Consider the links done during a delete-min, say of a node $u$.  Deletion of the original root converts the list of children of this root into a list of new roots.  Group the roots on this list, not including the last (rightmost) root, into consecutive sets of four.  Associate four links with each group: The two first-pass and one second-pass pairing link in which its nodes participate, and the first assembly link in which the winner of the second-pass pairing link participates.  At most four new roots are not in a group. There are at most three links between the nodes not in a group.  All other links are associated with a group.  We charge the links not associated with a group to the insertion of $u$, which is a temporary node.  A~temporary node is charged in this way at most once.  We charge the links associated with a group in one of two ways.  If the second-pass assembly link of a group is real, we charge all four links of the group to this link.  If not, at least one of the two first-pass pairing links must be between two permanent nodes.  Such a link is never cut.  In this case, we charge all four links of the group to the insertion of the permanent node that lost the first-pass pairing link to another permanent node.  A~permanent node is charged in this way at most once.  Adding up the charges gives the lemma.
\end{proof}

An analog of Lemma~\cref{L:double-link-bound} holds for pairing heaps~\cite{SinnamonT-pairing}, but the bound is weaker: The total number of links is at most a constant times the number of insertions plus the number of real pairing links plus the number of real assembly links.  We do not know how to get a bound on the number of real assembly links small enough to prove that pairing heaps have the strict working-set bound. 

To bound the number of real second-pass pairing links, we use a standard potential-based argument~\cite{Tarjan85}.  We assign to each state of the data structure a real-valued \emph{potential}.  We define the \emph{amortized cost} of an operation to be its actual cost plus the net increase in potential caused by the effect of the operation on the state of the data structure.  Given a sequence of operations, the sum of their amortized costs equals the sum of their actual costs plus the final potential minus the initial potential, since the net increases in potential form a telescoping sum.   In our application the initial and final potentials are zero, so the sum of the amortized costs equals the sum of the actual costs.

We define the actual cost of an insertion to be zero, and the actual cost of a delete-min to be the number of real second-pass pairing links it does.  We define the potential in a way similar to that used by \citet{iacono-pairing-heaps} to prove that pairing heaps have the intermediate working-set bound if they end up empty, which is in turn based on the potential used to prove that splay trees have a (differently defined) working-set bound~\cite{splay-trees}.  Recall from~\cref{S:heaps-working-set} that the \emph{strict working set} of a temporary node $x$ is the set of temporary nodes in the heap that were inserted at or after the time $x$ was inserted and before $x$ was deleted.  (That is, the strict working set of $x$ includes $x$.)  The \emph{strict working-set size} of $x$ is the size of its strict working set.  The \emph{strict intermediate working-set size} $w(x)$ of $x$ is the maximum number of nodes in its strict working set that are in the heap at the same time, a time between the times just after $x$ is inserted and just before $x$ is deleted, inclusive. 

We define the \emph{weight} of a temporary node $x$ to be $1/w(x)^2$.  Each permanent node has a weight of zero.  We define the \emph{size} of a node $x$ to be the sum of the weights of all descendants of $x$, including $x$.  It follows from this definition that the size of a node is the sum of its weight and the sizes of all its children.  If $x$ is a child, we define the \emph{mass} $m(x)$ of $x$ to be the size of its parent just after $x$ lost its link with this parent.  Equivalently, this is the weight of the parent plus the sum of the sizes of $x$ and of all its siblings after it on the list of its siblings.  Roots have no mass except in the middle of a delete-min.  If $x$ is a root in the middle of a delete-min, $m(x)$ is the weight of the just-removed root plus the sum of the sizes of $x$ and of all roots after $x$ on the root list.  Finally, we define the \emph{potential} $\Phi(x)$ of a temporary node $x$ to be $(\log m(x))/2$, and the potential of a heap to be the sum of the potentials of its temporary nodes.

Given an arbitrary sequence of heap operations on an initially empty heap, the initial and final potentials are zero, the latter because all temporary nodes have been deleted.  Hence the sum of the amortized costs of the operations equals the total number of real second-pass pairing links.

In the rest of the proof, we use $c = \sum_{i=1}^\infty{1/i^2}$. Note that $c$ is a well-defined constant since the sum on the right-hand side converges. 

\begin{lemma}\label{L:double-size-bound}
At all times, the size and mass of any heap node is at most $c$.
\end{lemma}
\begin{proof}
Let $x_1, x_2,\dots, x_k$ be the temporary nodes in a heap at some time, ordered in decreasing order by insertion time ($x_1$ last, $x_k$ first).  For each $i$, we have $w(i) \geq i$.  The lemma follows.   
\end{proof}

\begin{lemma}\label{L:double-amortized-cost}
In a double pairing heap, the amortized cost of an insertion is $\OO(1)$, and the amortized cost of a delete-min of item $u$ is $\OO(1+\log w(u))$. 
\end{lemma}

\begin{proof}
An insertion has an actual cost of zero.  An insertion into an empty heap produces a heap with zero potential.  Consider an insertion of $x$ into a heap with root $y$.  If $x$ is permanent, its insertion does not change the potential of the heap.  Suppose $x$ is temporary.  If $y$ is permanent, $x$ wins the link with $y$, leaving the potential of the heap unchanged.  Suppose $y$ is temporary.  If $x$ wins the link with $y$, $y$ acquires a potential of at most $(\log c)/2$ by \cref{L:double-size-bound}.  If $x$ loses the link with $y$, $x$ acquires a potential of at most $(\log c)/2$ by \cref{L:double-size-bound}.  In each case, the insertion increases the potential by at most $(\log c)/2$.  Hence the amortized cost of the insertion is $\OO(1)$.

The analysis of deletion is considerably more complicated.  Consider a delete-min that deletes root~$u$.  First, we bound the change in potential caused by the first pairing pass.  Then we obtain a combined bound on the amortized cost of the second pairing pass plus that of the assembly pass.

If $x$ and $y$ are temporary nodes on the root list with $x$ preceding $y$, then $m(x) \geq m(y)$.  Consider a pair of nodes $x$ and $y$ linked during the first pairing pass, with $x$ preceding $y$ on the root list before the link.  If either $x$ or $y$ is permanent, the link does not change the potential.  Suppose $x$ and $y$ are both temporary.  Let unprimed and primed variables denote their values just before and just after the link, respectively.  If $x$ wins the link, $m'(x)= m(x)$ and $m'(y) \leq m(x)$.  If $y$ wins the link, $m'(y)=m(x)$ and $m'(x) \leq m(x)$.  In either case, the potential increase caused by the link is at most $\Phi(x) - \Phi(y)$.  It follows that the sum of the increases in potential caused by the first-pass pairing links is bounded by a telescoping sum that totals $\Phi(v')-\Phi(v)$, where $v'$ is the first (leftmost) temporary root on the list of roots and $v$ is the last (rightmost).  This difference is at most $(\log c - \log(1/w(u)^2))/2 = \OO(1+\log w(u))$, because the weight of the deleted root $u$ is included in the mass of $v$.

For the purpose of analyzing the second pairing pass and the assembly pass, we change the order of the links.  This does not affect the final tree and hence does not affect the total amortized cost of the two passes.  We do the pairing links in last-to-first order rather than first-to-last, and we do each assembly link as soon as both nodes to be linked have won their second-pass pairing link, if any.  Specifically, if the second pairing pass links the last two links on the root list, we begin by linking these roots.  This is a \emph{zig step}.  There is at most one zig step.  Then we repeat the following \emph{zig-zig} step until only one root remains: Let $x$, $y$, and $z$ be the last three roots on the root list, with $x$ preceding $y$ and $y$ preceding $z$.  Link $x$ and $y$ by a pairing link, and then link the winner with $z$ by an assembly link.  After these two links, the winner of the second link takes the place of $x$, $y$, and $z$ on the root list and becomes the last root on the root list.

We shall bound the amortized cost of each step by a difference of node potentials.  These form a telescoping sum whose total gives us the desired bound.  The first term in this sum has an extra ``$+1$'' that counts one link whose cost cannot be folded into the corresponding difference in potentials.  This analysis is exactly like that of splay trees~\cite{splay-trees}, made simpler by the lack of a zig-zag case but made more complicated by the effect of permanent nodes.

Phantom links have no cost and do not change the potential.  After a step in which a temporary node participates, the rightmost root is a temporary node, and this remains true subsequently.

In analyzing steps, we shall denote by $z$ and $z'$ and by $\Phi$ and $\Phi'$ the rightmost root and the potential function before and after the step, respectively.

Consider the first step in which a temporary node participates.  Suppose this is a zig step (the only zig step).  There are two cases. If the link is between a temporary node and a permanent node, then the link costs nothing and does not change the potential.  The resulting potential change is at most $0 \leq \Phi'(z') - \Phi(v)$, where $v$ is the rightmost temporary node before the link.  If the link is between two temporary nodes, the resulting potential change is at most $\Phi'(z') - \Phi(v)$, where $v$ is the rightmost temporary node before the link.  In either case the amortized cost of the step, including the cost of the link if it is real, is at most $3(\Phi'(z') - \Phi(v))+1$.

Suppose the first step in which a temporary node participates is a zig-zig step.  Let $x$, $y$, and $z$ be the three rightmost roots just before the step.  Suppose $z$ is permanent.  Then the assembly link in the step is phantom and does not change the potential.  The argument in the previous paragraph shows that the amortized cost of the step is at most $3(\Phi'(z') - \Phi(v))+1$.

The remaining possibility is that the first step in which a temporary node participates is a zig-zig step in which the rightmost root is temporary.  We treat this as an instance of the general case, which is \emph{any} zig-zig step in which the rightmost root $z$ is temporary.  Consider such a step.  Let $x$, $y$, and $z$ be the three rightmost roots before the step, with $y$ left of $z$ and $x$ left of $y$.  If $x$ or $y$ is permanent, the second-pass pairing link done in the step is phantom and does not count.  The amortized cost of the step is the increase in potential, which is at most $3(\Phi'(z') - \Phi(z))$ by an argument like those in the previous cases.  

The final possibility is that $x$, $y$, and $z$ are all temporary.  This is the heart of the argument.  We fold the cost of the real pairing link between $x$ and $y$ into the potential difference bounding the amortized cost of the step.  The analysis is essentially the same as that used to analyze the zig-zig case in splaying~\cite{splay-trees}, subsequently used in the original analysis of pairing heaps~\cite{pairing-heaps}.

Let $x'$ be the loser of the pairing link and $y'$ the loser of the assembly link.  Node $z'$, the rightmost root after the step, is the winner of the assembly link.  Nodes $x'$, $y'$, and $z'$ are some permutation of $x$, $y$, and $z$ such that $x'$ is either $x$ or $y$.  Let $m$ and $m'$ denote the node masses before and after the step, respectively.  Then $\Phi(x)=\Phi'(z') \geq \Phi'(y')$ and $\Phi(y) \geq \Phi(z)$.  Also, $\Phi'(x')+ \Phi(z) \leq 2\Phi'(z') - 1$. This follows from two inequalities: $m'(z') = s(x)+s(y)+m(z) \geq m(z) + m'(x')$, and 
$\log(a+b) - 1 \geq (\log a + \log b)/2$ for any positive $a$ and $b$.  These combine to give $\log m'(z') - 1 \geq (\log m(z) +\log m'(x'))/2$, which can be rewritten as $2\Phi'(z') - \Phi(z) - 1 \geq \Phi'(x')$.  The net increase in potential caused by the two links is thus 
\begin{align*}
&\Phi'(x')+\Phi'(y')+\Phi'(z')-\Phi(x)-\Phi(y)-\Phi(z)\\
&=\Phi'(x')+\Phi'(y')-\Phi(y)-\Phi(z) \\ 
	&\leq \Phi'(x') + \Phi'(z') - \Phi(z) - \Phi(z) && [\,\Phi'(y') \le \Phi'(z'), \Phi(y) \ge \Phi(z)\,]  \\
	&\leq 3\Phi'(z')-3\Phi(z) - 1 && [\,2\Phi'(z')-\Phi(z)- 1\geq \Phi'(x')]
\end{align*}
The amortized cost of the step, including the unit cost of the real pairing link, is at most $3(\Phi'(z')-\Phi(z))$.  Node $z'$ becomes node $z$ in the next step.

Summing the amortized cost bound over all steps, the sum telescopes and totals $3(\Phi'(v')-\Phi(v)) + 1$, where $v$ is the last temporary root on the root list just after deletion of the original root, $v'$ is the final remaining root, and $\Phi$ and $\Phi'$ are the node potentials before and after all the links, respectively.  Since the potential of $v'$ after the links is at most $(\log c)/2$, and the potential of $v$ before any of the links is at least $-2\log w(u)$, the amortized cost of the combined second pairing and assembly passes is $\OO(1+\log w(u))$.

After the assembly pass, the final root $v'$ includes in its mass the weight of the deleted root $u$.  After the deletion, $v'$ has no mass and hence no potential.  At the end of the assembly pass, the potential of $v'$ is at least $\log(1/w(u)^2)$.  Hence the increase in potential of $v'$ from the end of the assembly pass to the end of the deletion is at most $2\log(w(u))$.

Combining our bounds, we find that the amortized cost of deleting root $u$ is $\OO(1+ \log w(u))$.
\end{proof}

\cref{L:double-link-bound,L:double-amortized-cost} give us the following theorem:

\begin{theorem}
An arbitrary sequence of operations on an initially empty double pairing heap takes $\OO(1)$ worst-case time and no comparisons for each find-min, $\OO(1)$ amortized time and comparisons for each insertion, and $\OO(1 + \log w(u))$ amortized time and comparisons to do a delete-min that deletes item $u$ having strict intermediate working-set size $w(u)$.
\end{theorem}

\section{An Alternative Algorithm}\label{S:alternative}
We briefly discuss the alternative algorithm by~\citet{vanderhoog2024}.
Their idea is as follows:  Find a longest path $P$.  Initialize a list $L$ to contain the vertices on $P$ in order.  Delete the vertices on $P$ from $G$ to form graph $G\setminus P$.  Insert the vertices not on $P$ into $L$ one-by-one in any topological order of $G\setminus P$.  To insert a vertex $v$ into $L$, do an exponential/binary search in $L$ starting from the largest vertex $u$ such that $uv$ is an arc, or from the beginning of $L$ if $v$ is an original source of $G\setminus P$.  

With appropriate data structures, this algorithm has the same asymptotic efficiency as topological heapsort with lookahead: It runs in $\OO(m+\log T)$ time and does $\OO(\log T)$ comparisons.  To support the exponential/binary searches in $L$, the authors represent $L$ by a homogeneous finger search tree~\cite{HuddlestonM82}.  In addition, they need a data structure to find the largest vertex $u$ such that $uv$ is an arc, for each vertex $v$ to be inserted into $L$ (without doing additional vertex comparisons).  This problem requires doing order queries in a list subject to insertions.  There are data structures for this problem that support order queries and insertions in $\OO(1)$ time~\cite{Tsakalidis83,DietzS87,BenderCDFZ02}.

In comparison, our algorithm uses just one array and a standard heap implementation, specifically a pairing heap.

A limitation of the algorithm of~\cite{vanderhoog2024} is that it does not find the vertices in increasing sorted order in a time bound that depends on the number so far found.  As a result, their approach does not seem to extend to the top-$k$ DAG sorting problem.
\end{document}